%% file: arxiv.tex
\documentclass{article}

\usepackage{arxiv}

\input{definitions}

\usepackage[utf8]{inputenc} 
\usepackage[T1]{fontenc}    
\usepackage{hyperref}       
\usepackage{url}            
\usepackage{booktabs}       
\usepackage{amsfonts}       
\usepackage{nicefrac}       
\usepackage{microtype}      
\usepackage{lipsum}
\usepackage{graphicx}
\graphicspath{ {./images/} }

\title{Unmatched Control Barrier Functions: \\ Certainty Equivalence Adaptive Safety}

\author{
 Brett T. Lopez \\
  Dept. of Mechanical and Aerospace Engineering\\
  University of California, Los Angeles\\
  Los Angeles, CA 90095 \\
  \texttt{btlopez@ucla.edu} \\
   \And
 Jean-Jacques Slotine \\
  Dept. of Mechanical Engineering\\
  Massachusetts Institute of Technology\\
  Cambridge, MA 02139 \\
  \texttt{jjs@mit.edu} \\
}

\begin{document}
\maketitle
\begin{abstract}
This work applies universal adaptive control to control barrier functions to achieve forward invariance of a safe set despite the presence of unmatched parametric uncertainties.
The approach combines two ideas.
The first is to construct a family of control barrier functions that ensures the system is safe for all possible models.
The second is to use online parameter adaptation to methodically select a control barrier function and corresponding safety controller from the allowable set.
While such a combination does not necessarily yield forward invariance without additional requirements on the barrier function, we show that such invariance can be established by simply adjusting the adaptation gain online.
It is also shown that the developed method is applicable to systems with safety constraints that have a relative degree greater than one.
This work thus represents the first adaptive safety approach that successfully employs the certainty equivalence principle for general state constraints without sacrificing safety guarantees.
\end{abstract}


\section{Introduction}

Safety-critical controllers rely on precise model knowledge to ensure forward invariance of a safe set.
Although these controllers often possess some inherent robustness \cite{xu2015robustness}, techniques that guarantee safety while effectively compensating for model uncertainties with minimal conservatism have only recently been proposed.
One such framework is based on the notion of adaptive safety; a paradigm that achieves forward invariance of a safe set using results from adaptive control and model estimation theory. 
The work by \cite{taylor2020adaptive} introduced the adaptive control barrier function (aCBF) --- analogous to an adaptive control Lyapunov function in adaptive control \cite{krstic1995control} --- and showed that subsets of a safe set were forward invariant when an aCBF was used to construct controllers. 
The conservatism of an aCBF was addressed by \cite{lopez2020robust} through the so-called robust adaptive control barrier function (RaCBF).
There are two key differences between an aCBF and RaCBF.
Firstly, an RaCBF yields less conservative controllers as the system is allowed to approach the boundary of the safe set.
Secondly, an RaCBF can be combined with model estimation to further reduce conservatism if monotonic reduction in the model uncertainty can be established.
Due to its effectiveness and strong theoretical guarantees, the adaptive safety paradigm has seen several extensions by others in the controls and robotics communities (see, e.g., \cite{maghenem2021adaptive,isaly2021adaptive,black2021fixed,cohen2022high}).

A fundamental limitation of current adaptive safety approaches is the inability to employ the certainty equivalence principle when the barrier function depends on unknown model parameters.
Model parameter dependency usually arises when either 1) the system possess unmatched model uncertainties, i.e., unknown dynamics outside the span of the control input matrix, or 2) the uncertainty enters through the control input matrix for systems with control constraints. 
The certainty equivalency principle is a design philosophy that entails constructing barrier functions as if the model were known. 
Then, the uncertain parameters can simply be replaced with their online estimates.
This philosophy works seamlessly when the barrier function does not depend on unknown parameters.
Conversely, it is very difficult to establish forward invariance of a parameter-dependent barrier function as sign-indefinite terms arise in the forward invariance proof.
To cancel out the sign-indefinite terms, \cite{taylor2020adaptive,lopez2020robust} construct a barrier function for a \emph{modified} system that depends on its own (unknown) barrier function.
Generating such a barrier function is generally no easier than dealing with the problematic terms directly.

The main contribution of this work is a stable adaptive safety algorithm which employs the certainty equivalence principle to achieve set invariance through online parameter adaptation. 
Unlike previous works, one just needs to construct a family of control barrier functions for all possible models --- a much simpler procedure than that originally proposed by \cite{taylor2020adaptive,lopez2020robust}. 
It is shown that a safe set can be made forward invariant with a novel direct adaptation law that systematically adjusts the adaptation gain online; a novel technique recently developed for adaptive control with unmatched uncertainties \cite{lopez2021universal} that has also found uses in other areas of control, e.g., direct adaptive optimal control \cite{lopez2022adaptive}.
Furthermore, the direct adaptation law can be combined with model learning to improve parameter adaptation transients and reduce conservatism.
The derived adaptation laws can also be used with safety constraints more than one derivative away from the input.
As a result, this is the first work to successfully utilize the certainty equivalence principle in adaptive safety, and represents an important step towards safety-critical control of uncertain systems.

\paragraph{Notation:} The set of positive and strictly-positive scalars will be denoted as $\mathbb{R}_+$ and $\mathbb{R}_{>0}$, respectively.
The shorthand notation for a function $T$ parameterized by a vector $a$ with vector argument $s$ will be $T_a(s) \triangleq T(s;a)$.
The partial differentiation with respect to variable $x \in \mathbb{R}^n$ of function $N(x,y)$ will be $\nabla_x N(x,y) = \partial N / \partial x \in \mathbb{R}^{n}$.
The subscript for $\nabla$ will be omitted when it is clear which variable the differentiation is with respect to.
The Lie derivative of a scalar function $h:\mathbb{R}^n \rightarrow \mathbb{R}$ with respect to a vector field $f : \mathbb{R}^n \rightarrow \mathbb{R}^n$ is $L_f h \triangleq \nabla h^\top f$.

\section{Problem Formulation}
Consider the uncertain nonlinear system
\begin{equation}
\label{eq:dyn}
    \dot{x} = f(x) -\Delta(x)^\top \theta + g(x) u,
\end{equation}
with state $x\in\mathbb{R}^n$, control input $u \in \mathbb{U} \subseteq \mathbb{R}^m$, known dynamics $f:\mathbb{R}^n \rightarrow \mathbb{R}^n$, unknown parameters $\theta \in \Theta \subset \mathbb{R}^p$ with known regressor $\Delta: \mathbb{R}^n \rightarrow \mathbb{R}^{p \times n}$, and known control input matrix $g: \mathbb{R}^n \rightarrow \mathbb{R}^{n\times m}$.
In this work we derive an adaptive safety controller that ensures $x(t) \in \mathcal{C}$ for all time where $\mathcal{C}$ is a set of safe states. 
The following assumption is made on the unknown parameters $\theta$.

\begin{assumption}
\label{assumption:params}
The unknown parameters $\theta$ belong to a known closed and bounded set $\Theta \subset \mathbb{R}^p$.
\end{assumption}

An immediate consequence of Assumption~\ref{assumption:params} is that the parameter estimation error $\tilde{\theta} \triangleq \hat{\theta} - \theta$ must also belong to a known closed and bounded set, i.e., $\tilde{\theta} \in \tilde{\Theta}$.
Moreover, each element must then have a finite supremum where $\tilde{\vartheta}_i \triangleq \underset{\tilde{\theta}_i  \in  \tilde{\Theta}}{\mathrm{sup}} ~ \tilde{\theta}_i$ for $i=1,\dots,p$.
Note that $\tilde{\vartheta}$ is equivalent to the maximum possible parameter estimation error.

\section{Background: Adaptive Safety}

The following definitions are stated for completeness, a more thorough treatment can be found in \cite{ames2016control,ames2019control} and references therein.
See \cite{taylor2020adaptive,lopez2020robust} for the first works on adaptive safety.

Let the set $\mathcal{C} \subset \mathbb{R}^n$ be a 0-superlevel set of a continuously differentiable function $h: \mathbb{R}^n \rightarrow \mathbb{R}$ where

\begin{equation*}
    \begin{aligned}
    \mathcal{C} & = \left\{ x \in \mathbb{R}^n : h(x) \geq 0  \right\}, \\
    \partial \mathcal{C} & = \left\{ x \in \mathbb{R}^n : h(x) = 0  \right\}, \\
    \text{Int}\left(\mathcal{C}\right) & = \left\{ x \in \mathbb{R}^n : h(x) > 0  \right\}.
    \end{aligned}
\end{equation*}

The following definitions assume the dynamics \cref{eq:dyn} are Lipschitz (at least locally) so that there exists a unique solution $x(t)$ for $t \geq t_0$ with initial condition $x_0 \triangleq x(t_0)$.

\begin{definition}
\label{def:fi}
The set $\mathcal{C}$ is \emph{forward invariant} if for every $x_0\in \mathcal{C}$, $x(t) \in \mathcal{C}$ for all $t \geq t_0$.
\end{definition}

\begin{definition}
\label{def:safety}
A system is \emph{safe} with respect to set $\mathcal{C}$ if the set $\mathcal{C}$ is forward invariant.
\end{definition}

\begin{definition}
A continuous function $\alpha : \mathbb{R} \rightarrow \mathbb{R}$ is an \emph{extended class $\mathcal{K}_\infty$ function} if it is strictly increasing, $\alpha(0) = 0$, and is defined on the entire real line.
\end{definition}

When model uncertainty is present, as is the case in \cref{eq:dyn}, it is challenging or infeasible to derive a controller that renders a safe set forward invariant.
Conceptually, adaptive safety is a framework that uses tools from adaptive control theory to systematically compute a safe controller via online parameter adaptation.
Central to adaptive safety is the notion of model-parameterized safe sets $\mathcal{C}_{\theta}$ which are shown to be forward invariant with an aCBF \cite{taylor2020adaptive} or RaCBF \cite{lopez2020robust}.
Due to the similarities between \cite{lopez2020robust} and this work, only the core results from \cite{lopez2020robust} are summarized below.

\begin{definition}[Robust Adaptive Control Barrier Function \cite{lopez2020robust}]
\label{def:racbf}
Let $\mathcal{C}_\theta$ be a family of 0-superlevel sets parameterized by $\theta$ for a continuously differentiable function $h^r_\theta:\mathbb{R}^n\times \mathbb{R}^p\rightarrow\mathbb{R}$.
The function $h^r_\theta(x)$ is a \emph{robust adaptive control barrier function} (RaCBF) if there exists a controller $u \in \mathbb{U}$ and extended class $\mathcal{K}_\infty$ function $\alpha(\cdot)$ such that for every $\theta\in\Theta$
\begin{equation}
\label{eq:rcbfu}
    \underset{u  \in  \mathbb{U}}{\text{sup}}~ \left\{ \nabla_x h^r_\theta(x)^\top \left[ f(x) - \Delta(x)^\top \Lambda_{\theta}(x) + g(x) u\right]\right\} \geq - \alpha(h^r_\theta(x)),
\end{equation}
where $\Lambda_{\theta}(x) \triangleq \theta - \gamma \nabla_\theta h^r_{\theta}(x)$ and $\gamma$ is an admissible adaptation gain.
\end{definition}

Observe that an RaCBF is designed for a modified system that depends upon its own RaCBF via $\Lambda_\theta(x)$.
This is the main drawback of using an RaCBF (and an aCBF as the same modification is employed).
Fundamentally, the modified dynamics are a byproduct of attempting to use the certainty equivalence principle with model-parameterized safe sets. 
It can be shown that the extra term is actually related to the adaptation law derived for an RaCBF \cite{lopez2020robust}.
Therefore, an RaCBF is constructed to account for parameter adaptation transients and hence represents a departure from the certainty equivalence principle design philosophy. 
In this work we will show that a true certainty equivalence adaptive safety framework is possible when the adaptation gain is adjusted online \cite{lopez2021universal}.

The main theorem from \cite{lopez2020robust} is stated below and will serve as a useful comparison for results derived in \cref{sec:results}.

\begin{theorem}[\cite{lopez2020robust}]
\label{thm:racbf}
Let $\mathcal{C}_{\hat{\theta}}$ be a 0-superlevel set of a continuously differentiable function $h^r_{\hat{\theta}}:\mathbb{R}^n \times \mathbb{R}^p\rightarrow\mathbb{R}$. If $h^r_{\hat{\theta}}(x)$  is an RaCBF on $\mathcal{C}^r_{\hat{\theta}} \triangleq \{x \in \mathbb{R}^n,  \hat{\theta} \in \Theta : h^r_{\hat{\theta}}(x) \geq \tfrac{1}{2\gamma} \tilde{\vartheta}^\top \tilde{\vartheta}\} \subseteq \mathcal{C}_{\hat{\theta}}$ then any locally Lipschitz continuous controller satisfying
\begin{equation}
\label{eq:racbfu}
    \underset{u  \in  \mathbb{U}}{\mathrm{sup}}~ \left\{ \nabla_x h^r_{\hat{\theta}}(x)^\top \left[ f(x) - \Delta(x)^\top \Lambda_{\hat{\theta}}(x) + g(x) u\right]\right\} \geq - \alpha\left(h^r_{\hat{\theta}}(x) - \tfrac{1}{2\gamma} \tilde{\vartheta}^\top \tilde{\vartheta}\right)
\end{equation}
renders $\mathcal{C}_{\hat{\theta}}$ safe with the adaptation law
\begin{equation}
\label{eq:radapt}
    \dot{\hat{\theta}} = \gamma \Delta(x) \nabla_x h^r_{\hat{\theta}}(x) \\
\end{equation}
where $\gamma$ is an admissible adaptation gain.
\end{theorem}

\section{Main Results}
\label{sec:results}

\subsection{Overview}

This section contains the main results of this work. 
First, it is shown that set invariance is achieved by combining the so-called \emph{unmatched control barrier function} and direct adaptive control with online adaptation gain adjustment.
Then, several modifications to the direct adaptive control algorithm that can improve transients and reduce conservatism are discussed.
Finally, a unified adaptive safety tracking min-norm controller is presented.

The subsequent results will make use of a special class functions called scaling functions.

\begin{definition}[Scaling Function]
\label{def:scaling}
A \emph{scaling function} $v : \mathbb{R} \rightarrow \mathbb{R}$ satisfies the following conditions for $\zeta > 1$ and $\rho \in \mathcal{R} \subset \mathbb{R}$
\begin{gather}
         1 \leq v(\rho) \leq \zeta < \infty,  \tag{C1} \\
         \nabla v(\rho) > 0. \tag{C2}
\end{gather}
\end{definition}

\begin{remark}
    One example of a suitable scaling function is  $v(\rho) = \arctan(\rho) + 1$  where  $\rho \in [0~10]$. 
\end{remark}

\subsection{Direct Adaptive Safety}

We will consider two safe sets defined by a continuously differentiable function $h_\theta:\mathbb{R}^n\times \mathbb{R}^p\rightarrow\mathbb{R}$, namely  $\mathcal{C}_\theta \triangleq \left\{ x \in \mathbb{R}^n,  \theta \in \Theta : h_\theta(x) \geq 0  \right\}$  and  $\mathcal{C}^r_\theta \triangleq \{ x \in \mathbb{R}^n,  \theta \in \Theta : h_\theta(x) \geq \tfrac{1}{2\gamma} \tilde{\vartheta}^\top \tilde{\vartheta}\} $ where $ \mathcal{C}^r_\theta \subseteq \mathcal{C}_\theta$.

\begin{definition}[Unmatched Control Barrier Function]
\label{def:ucbf}
Let $\mathcal{C}_\theta$ be a family of 0-superlevel sets parameterized by $\theta$ for a continuously differentiable function $h_\theta:\mathbb{R}^n\times \mathbb{R}^p\rightarrow\mathbb{R}$.
The function $h_\theta(x)$ is an \emph{unmatched control barrier function} (uCBF) if there exists a controller $u \in \mathbb{U}$ and extended class $\mathcal{K}_\infty$ function $\alpha(\cdot)$ such that for every $\theta\in\Theta$
\begin{equation}
\label{eq:ucbfu}
    \underset{u  \in  \mathbb{U}}{\text{sup}}~ \left\{ \nabla_x h_\theta(x)^\top \left[ f(x) - \Delta(x)^\top \theta + g(x) u\right]\right\} \geq - \alpha(h_\theta(x)).
\end{equation}
\end{definition}

\begin{remark}
    In addition to  $\alpha(\cdot) \in \mathcal{K}_\infty$, we require that  $c \alpha(r) \leq \alpha(cr)$  for  $c\geq 1$.
    This property is not restrictive as it satisfied by many common choices for  $\alpha(\cdot)$  found in the literature. 
\end{remark}

Fundamentally, condition \cref{eq:ucbfu} states that there exists a controller that renders $\mathcal{C}_\theta$ invariant for every $\theta \in \Theta$.
Or, put another way, the uncertain system \cref{eq:dyn} can be made safe for every $\theta \in \Theta$.
This is analogous to an uncertain system being stabilizable for every $\theta \in \Theta$ in the context of adaptive control, as discussed in \cite{lopez2021universal}.
Note that \cref{eq:ucbfu} is an invariance condition for $\mathcal{C}_\theta$ with the \emph{actual} dynamics as opposed to the modified dynamics used in \cite{taylor2020adaptive,lopez2020robust}.
This distinction has both theoretical and practical implications.
In particular, safety is now formulated as an inherent property of the system since the actual dynamics are being evaluated for safety. 
Moreover, in terms of constructing an uCBF, the safety condition \cref{eq:ucbfu} is much easier to verify as it preservers bilinearity of $h_{\theta}(x)$ and $u$ \cite{ames2019control}.
Conversely, the condition for an aCBF or RaCBF is nonconvex so systematically constructing either barrier function is more difficult.

The following theorem establishes forward invariance of a parameter-dependent safe set $\mathcal{C}_{\theta}$ when \cref{def:ucbf} is combined with direct parameter adaption and online adjustment of the adaptation gain.

\begin{theorem}
\label{thm:direct}
Let $\mathcal{C}_{\hat{\theta}}$ be a 0-superlevel set of a continuously differentiable function $h_{\hat{\theta}}:\mathbb{R}^n \times \mathbb{R}^p\rightarrow\mathbb{R}$. If $h_{\hat{\theta}}(x)$  is an uCBF on $\mathcal{C}^r_{\hat{\theta}} \triangleq \{x \in \mathbb{R}^n,  \hat{\theta} \in \Theta : h_{\hat{\theta}}(x) \geq \tfrac{1}{2\gamma} \tilde{\vartheta}^\top \tilde{\vartheta}\} \subseteq \mathcal{C}_{\hat{\theta}}$ then any locally Lipschitz continuous controller satisfying
\begin{equation}
\label{eq:rucbfu}
    \underset{u  \in  \mathbb{U}}{\mathrm{sup}}~ \left\{ \nabla_x h_{\hat{\theta}}(x)^\top \left[ f(x) - \Delta(x)^\top \hat{\theta} + g(x) u\right]\right\} \geq - \alpha\left(h_{\hat{\theta}}(x) - \tfrac{1}{2\gamma} \tilde{\vartheta}^\top \tilde{\vartheta}\right).
\end{equation}
renders $\mathcal{C}_{\hat{\theta}}$ safe with the \emph{direct adaptation law}
\begin{subequations}
\label{eq:adapt}
    \begin{align}
        \dot{\hat{\theta}} &= \gamma v(\rho) \Delta(x) \nabla_x h_{\hat{\theta}}(x) \label{eq:theta_direct} \\
        \dot{\rho} &= - \frac{v(\rho)}{\nabla v(\rho)} \frac{1}{h_{\hat{\theta}}(x)+\eta} \nabla_{\hat{\theta}} h_{\hat{\theta}}(x)^\top  \dot{\hat{\theta}} \label{eq:rho_direct}
    \end{align}
\end{subequations}
where $\gamma$ is an admissible adaptation gain, $v(\rho)$ is a scaling function, and $\eta \in \mathbb{R}_{>0}$ is a design parameter.
\end{theorem}

\begin{proof}
Consider the barrier-like function 
\begin{equation*}
    h(t) = v(\rho) \left(h_{\hat{\theta}}(x) + \eta \right) - \tfrac{1}{2\gamma} \tilde{\theta}^\top \tilde{\theta},
\end{equation*}
where $\eta > 0$.
We will show that the adaptation law \cref{eq:adapt} yields $h_{\hat{\theta}}(t) \geq 0$ for all $t$ which is equivalent to $\mathcal{C}_{\hat{\theta}}$ being forward invariant.
Differentiating $h(t)$ and applying \cref{eq:adapt} yields
\begin{equation*}
    \begin{aligned}
        \dot{h}(t)  &= v(\rho) \nabla_x h_{\hat{\theta}}(x)^\top \left[ f(x) - \Delta(x)^\top \theta + g(x) u \right] + v(\rho) \nabla_{\hat{\theta}} h_{\hat{\theta}}(x)^\top \dot{\hat{\theta}} + \dot{\rho} \nabla v(\rho) \left(h_{\hat{\theta}}(x)+\eta\right) - \tfrac{1}{\gamma}\tilde{\theta}^\top\dot{\hat{\theta}} \\
        & = v(\rho) \nabla_x h_{\hat{\theta}}(x)^\top \left[ f(x) - \Delta(x)^\top \hat{\theta} + \Delta(x)^\top \tilde{\theta} + g(x) u \right] + v(\rho) \nabla_{\hat{\theta}} h_{\hat{\theta}}(x)^\top \dot{\hat{\theta}} + \dot{\rho}\nabla v(\rho) \left(h_{\hat{\theta}}(x) + \eta\right) - \tfrac{1}{\gamma}\tilde{\theta}^\top\dot{\hat{\theta}} \\
        & = v(\rho) \nabla_x h_{\hat{\theta}}(x)^\top \left[ f(x) - \Delta(x)^\top \hat{\theta} + g(x) u \right] \\ 
        & \geq - \alpha \left( v(\rho) h_{\hat{\theta}}(x) - v(\rho)\tfrac{1}{2\gamma} \tilde{\vartheta}^\top\tilde{\vartheta}\right),
    \end{aligned}
\end{equation*}
where the inequality arises from \cref{eq:rucbfu} and the property  $c\alpha(r) \leq \alpha(cr)$  for  $c\geq1$. 
Since  $|\tilde{\theta} | \leq \tilde{\vartheta}$  and  $v(\rho) \geq 1$  then
\begin{equation*}
    \begin{aligned}
        h(t) & \geq   v(\rho)(h_{\hat{\theta}}(x)+\eta) - \tfrac{1}{2\gamma}\tilde{\vartheta}^\top\tilde{\vartheta} \\
          & \geq v(\rho)\left(h_{\hat{\theta}}(x)+\eta - \tfrac{1}{2\gamma}\tilde{\vartheta}^\top\tilde{\vartheta}\right),
    \end{aligned}
\end{equation*}
yielding $\dot{h}(t) \geq - \alpha \left( h(t) - v(\rho)\eta \right)$ which implies  $h(t) \geq v(\rho) \eta > 0$  for all  $t \geq 0$  if $h(0) \geq v(\rho(0)) \eta$.
Since  $v(\rho) \eta \leq h(t) \leq v(\rho) \left( h_{\hat{\theta}}(x) + \eta\right)$, then  $h_{\hat{\theta}}(x(t)) \geq 0$  for all $t$.
Therefore, the controller \cref{eq:rucbfu} and direct adaptation law \cref{eq:adapt} render the set  $\mathcal{C}_{\hat{\theta}}$  forward invariant. 
\end{proof}

\vskip -0.1in

\begin{remark}
One notable modification to the adaptation law in \cref{eq:adapt} is the use of Bregman divergence to impose physical consistency \cite{wensing2017linear,lee2018natural} or sparsity \cite{ghai2020exponentiated,boffi2021higher} on the parameter estimates. 
See \cref{def:bregman} in the Appendix.
\end{remark}

It is instructive to analyze the online gain adjustment mechanism to develop an intuition about how the technique achieves forward invariance. 
In the case of no gain adjustment, i.e., $\dot{\rho} = 0$, then $\dot{h}(t)$ becomes
\begin{equation*}
    \dot{h}(t) \geq -\alpha(h(t) - v(\rho)\eta) + v(\rho) \nabla_{\hat{\theta}} h_{\hat{\theta}}(x)^\top \dot{\hat{\theta}}.
\end{equation*}
If $\nabla_{\hat{\theta}}h_{\hat{\theta}}(x)^\top \dot{\hat{\theta}} \geq 0$ then safety is preserved as the same inequality used to prove \cref{thm:direct} is obtained. 
However, if $\nabla_{\hat{\theta}}h_{\hat{\theta}}(x)^\top \dot{\hat{\theta}} < 0$ then safety might be compromised since this could lead to $h(t) < 0$ and subsequently $h_{\hat{\theta}}(x) < 0$.
From \cref{eq:adapt}, we see that $\dot{\rho}$ will be of opposite sign of $\nabla_{\hat{\theta}}h_{\hat{\theta}}(x)^\top \dot{\hat{\theta}}$.
Hence, if the parameter adaptation transients is negative (unsafe), then $\rho$ increases resulting in a larger effective adaptation gain $\gamma v(\rho)$.
Conversely, if the transients is positive (safe) then $\rho$ decreases yielding a smaller effective adaptation gain.
In this scenario, one could also set $\dot{\rho} = 0$ without sacrificing safety.
To summarize, the effective adaptation gain $\gamma v(\rho)$ will increase if the parameter adaptation transients compromises forward invariance of  $\mathcal{C}_{\hat{\theta}}$  while  $\gamma v(\rho)$  will decrease or remain constant if the parameter adaptation transients preserves safety. 
Note that the above analysis of adaptation gain adjustment is analogous to that in universal adaptive control \cite{lopez2021universal} where the effective adaptation gain changes to achieve a stable closed-loop system.

\cref{thm:direct} requires the adaptation gain $\gamma$ be admissible in order to prove forward invariance of $\mathcal{C}_{\hat{\theta}}$. 
The following corollary establishes a lower bound on $\gamma$ thereby making it admissible.

\begin{corollary}
An admissible adaptation gain $\gamma$ for the adaptation law in \cref{thm:direct} satisfies the inequality
\begin{equation}
\label{eq:rate}
    \gamma \geq \frac{\tilde{\vartheta}^\top \tilde{\vartheta}}{2  h_{\hat{\theta}}(x_0)},
\end{equation}
where $h_{\hat{\theta}}(x_0) \triangleq h(x(0),\hat{\theta}(0))$.
\end{corollary}
\begin{proof}
\cref{thm:direct} established $h_{\hat{\theta}}(x(t)) \in \mathcal{C}_{\hat{\theta}}$ uniformly if $h(0) \geq v(\rho_0) \eta$. We will show that this condition is satisfied with \cref{eq:rate}. Using the definition of $h(t)$,
\begin{equation*}
    \begin{aligned}
        h(0) & = v(\rho_0) \left(h_{\hat{\theta}}(x_0) + \eta \right) - \tfrac{1}{2\gamma} \tilde{\theta}_0^\top \tilde{\theta}_0 \\
        & \geq h_{\hat{\theta}}(x_0)+\eta - \tfrac{1}{2\gamma} \tilde{\vartheta}^\top \tilde{\vartheta}
    \end{aligned}
\end{equation*}
where we have chosen $v(\rho_0)=1$. Hence, $h(0) \geq v(\rho_0) \eta = \eta \iff \gamma \geq \frac{ \tilde{\vartheta}^\top \tilde{\vartheta}}{2 h_{\hat{\theta}}(x_0)}$, thus yielding \cref{eq:rate}.
\end{proof}

The lower bound for the adaptation gain is identical to that obtained in \cite{lopez2020robust} and similar to that in \cite{taylor2020adaptive} (in the latter case the initial parameter estimation error had to be known). 
Essentially, \cref{eq:rate} states that the closer $h_{\hat{\theta}}(x_0)$ is to $\partial \mathcal{C}_{\hat{\theta}}$ the faster the adaptation has to be in order to render $\mathcal{C}_{\hat{\theta}}$ invariant \cite{taylor2020adaptive,lopez2020robust}.

An interesting consequence of \cref{thm:direct} is that the set $\mathcal{C}^r_{\hat{\theta}} = \{ x \in \mathbb{R}^n,  \hat{\theta}\in\Theta: h_{\hat{\theta}}(x) \geq \tfrac{1}{2\gamma} \tilde{\vartheta}^\top \tilde{\vartheta} \}$ is input-to-state safe \cite{kolathaya2018input} with the proposed adaptive safety controller.

\begin{corollary}
\label{cor:issf}
The set $\mathcal{C}^r_{\hat{\theta}} = \{ x \in \mathbb{R}^n,  \hat{\theta}\in\Theta: h_{\hat{\theta}}(x) \geq \tfrac{1}{2\gamma} \tilde{\vartheta}^\top \tilde{\vartheta} \}$ is input-to-state safe (ISSf) with the controller and adaptation law in \cref{thm:direct}.
\end{corollary}
\begin{proof}
Follows immediately from the definition of ISSf which states that a set is ISSf if it is a subset of a forward invariant set. 
Since the controller and adaptation law render  $\mathcal{C}_{\hat{\theta}}$  invariant, and  $\mathcal{C}^r_{\hat{\theta}} \subseteq \mathcal{C}_{\hat{\theta}}$, then  $\mathcal{C}^r_{\hat{\theta}}$  is ISSf.
\end{proof}

\begin{remark}
Future work will investigate strengthening \cref{cor:issf} to show that the set  $\mathcal{C}^r_{\hat{\theta}}$  is asymptotically stable, as is the case with a RaCBF (see \cref{prop:asym_stable} in Appendix).  
\end{remark}

Depending on the choice of $v(\rho)$, one may need to modify \cref{eq:adapt} in order for  $\rho$  and  $v(\rho)$  to remain bounded.
One possibility is to reset  $\rho$  once it exceeds a certain threshold.
Even though  $u$,  $\hat{\theta}$,  and  $h_{\hat{\theta}}(x)$  remain continuous after a reset, a thorough analysis is required to ensure the closed-loop system remains safe despite the barrier-like function  $h(t)$  decreasing after the reset.
Alternatively, one could add damping to  $\dot{\rho}$,  thereby bounding $\rho$ but at the expense of rendering  $\mathcal{C}_{\hat{\theta}}$  ISSf, as shown in the following corollary.

\begin{corollary}
\label{corr:leak}
Let $\mathcal{C}_{\hat{\theta}}$ be a 0-superlevel set of a continuously differentiable function $h_{\hat{\theta}}:\mathbb{R}^n \times \mathbb{R}^p\rightarrow\mathbb{R}$. If $h_{\hat{\theta}}(x)$  is an uCBF on $\mathcal{C}^r_{\hat{\theta}} \triangleq \{x \in \mathbb{R}^n,  \hat{\theta} \in \Theta : h_{\hat{\theta}}(x) \geq \tfrac{1}{2\gamma} \tilde{\vartheta}^\top \tilde{\vartheta}\} \subseteq \mathcal{C}_{\hat{\theta}}$ then any locally Lipschitz continuous controller satisfying
\begin{equation}
\tag{\ref{eq:rucbfu}}
    \underset{u  \in  \mathbb{U}}{\mathrm{sup}}~ \left\{ \nabla_x h_{\hat{\theta}}(x)^\top \left[ f(x) - \Delta(x)^\top \hat{\theta} + g(x) u\right]\right\} \geq - \alpha\left(h_{\hat{\theta}}(x) - \tfrac{1}{2\gamma} \tilde{\vartheta}^\top \tilde{\vartheta}\right).
\end{equation}
renders $\mathcal{C}_{\hat{\theta}}$ input-to-state safe with the adaptation law
\begin{subequations}
\label{eq:adapt_leak}
    \begin{align}
        \dot{\hat{\theta}} &= \gamma v(\rho) \Delta(x) \nabla_x h_{\hat{\theta}}(x), \label{eq:theta_leak} \\
        \dot{\rho} &= \frac{v(\rho)}{\nabla v(\rho)} \frac{1}{h_{\hat{\theta}}(x)+\eta} \left[  -\sigma \rho + w_{\hat{\theta}}(x)  \right] , \label{eq:leak}
    \end{align}
\end{subequations}
where
\begin{equation}
\begin{aligned}
\label{eq:s_rho}
    w_{\hat{\theta}}(x) = \begin{cases} 0 ~ &\text{if} ~~ \nabla_{\hat{\theta}}h_{\hat{\theta}}(x)^\top  \dot{\hat{\theta}} \geq 0 \\
     - \zeta \nabla_{\hat{\theta}} h_{\hat{\theta}}(x)^\top \left[\gamma \Delta(x) \nabla_x h_{\hat{\theta}}(x)\right] ~~~~ &\text{otherwise},
    \end{cases}
\end{aligned}
\end{equation}

and $\gamma$ is an admissible adaptation gain, $v(\rho)$ is a scaling function, and $\eta,\sigma \in \mathbb{R}_{>0}$ are design parameters.
\end{corollary}
\begin{proof}
We must first establish that  $\rho$  is bounded from above and non-negative before showing  $\mathcal{C}_{\hat{\theta}}$  is ISSf.
Let  $d_{\hat{\theta}}(x,\rho) \triangleq \sigma \tfrac{v(\rho)}{\nabla v(\rho)} \tfrac{1}{h_{\hat{\theta}}(x)+\eta}$  which is strictly positive based on \cref{def:scaling,def:ucbf} and $\sigma, \eta \in \mathbb{R}_{>0}$.
Since \cref{eq:leak} is a stable\footnote{Stability can be established by forming the virtual system $\dot{y} = - d_{\hat{\theta}}(x,\rho) y + w_{\hat{\theta}}(x)$ where $d_{\hat{\theta}}(x,\rho)  > 0$ and $w_{\hat{\theta}}(x)$ is a bounded input. Note that  $ d_{\hat{\theta}}(x,\rho) \rho$  is commonly referred to as a \emph{leakage term}.} first order filter with a bounded input  $w_{\hat{\theta}}(x)$  (under the premise  $x,\hat{\theta}$  are bounded and  $\Delta(x),  h_{\hat{\theta}}(x)$  are continuously differentiable), then  $\rho$  must remain bounded.
Moreover, one can show that  $\rho(t) \geq 0$  for all  $t\geq 0$  if  $\rho(0) \geq 0$.
First consider the simple case where  $\nabla_{\hat{\theta}}h_{\hat{\theta}}(x)^\top  \dot{\hat{\theta}} \geq 0$. 
As noted above, this scenario yields the same forward invariance inequality used in the proof of \cref{thm:direct} so  $w_{\hat{\theta}}(x)$  can be trivially set to zero.
If  $\rho > 0$  then  $\rho \rightarrow 0$  exponentially with rate  $d(x,\rho)$  since  $w_{\hat{\theta}}(x) = 0$.
Conversely, if  $\nabla_{\hat{\theta}}h_{\hat{\theta}}(x)^\top  \dot{\hat{\theta}} < 0$  then $w_{\hat{\theta}}(x) + \nabla_{\hat{\theta}}h_{\hat{\theta}}(x)^\top  \dot{\hat{\theta}} > 0$  which implies that  $w_{\hat{\theta}}(x) > 0$  and subsequently  $\rho > 0$  if $\rho(0) \geq 0$.
Hence,  $\rho \geq 0$  and is bounded from above.

Following identical steps to those taken in the proof of \cref{thm:direct}, one obtains $\dot{h}(t) \geq - \alpha(h(t)-v(\rho)\eta) - \sigma v(\rho) \rho$ where  $\sigma v(\rho) \rho \geq 0$  and is bounded.
This yields
\begin{equation*}
    h(t) \geq v(\rho)\eta - \alpha^{-1}(\sigma v(\rho) \rho) \implies h_{\hat{\theta}}(x) \geq - \tfrac{1}{v(\rho)} \alpha^{-1} (\sigma v(\rho) \rho),
\end{equation*}
so the set  $\bar{\mathcal{C}}_{\hat{\theta}} \triangleq \{x\in \mathbb{R}^n,  \hat{\theta} \in \Theta : h_{\hat{\theta}}(x) + \tfrac{1}{v(\rho)} \alpha^{-1} (\sigma v(\rho) \rho) \geq 0 \}$  if forward invariant.
Moreover, since  $\bar{\mathcal{C}}_{\hat{\theta}} \supseteq \mathcal{C}_{\hat{\theta}} $ then  $\mathcal{C}_{\hat{\theta}}$  is input-to-state safe with \cref{eq:rucbfu,eq:adapt_leak,eq:s_rho}.
\end{proof}

\begin{remark}
\label{remark:rho_0}
A convenient byproduct of adding damping to $\dot{\rho}$ is that it naturally restores  $\rho$  to zero when the adaptation transients does not negatively impact safety.
If one designs  $v(\rho)$  so that  $v(0) = 1$  then the effective gain  $\gamma v(\rho)$  also returns to its nominal value  $\gamma$  without compromising safety.
\end{remark}

\begin{remark}
The proposed modification to the  $\rho$  dynamics is similar to the  $\sigma-$modification \cite{ioannou1984instability,ioannou1986robust}  and  $e$-modification \cite{narendra2012stable} used in adaptive control to improve robustness and transients.
Despite their similarities, this type of modification for online adjustment of the adaptation gain is quite novel.
Moreover, as discussed in \cref{remark:rho_0}, it is beneficial for $\rho\rightarrow0$ which is considered a detrimental behavior for the parameter estimates as they unlearn the values that yielded small tracking error \cite{narendra2012stable}.
\end{remark}

\subsection{Composite Adaptive Safety}

The direct adaptive safety controller in \cref{thm:direct} guarantees the safe set defined by  $\mathcal{C}_\theta$  is forward invariant for all possible models.
Better parameter adaptation transients can be obtained by combining \cref{eq:adapt} with a model estimator. 
There are a plethora of suitable model estimators that can be used with \cref{eq:adapt}.
One example of an effective and simple estimator is the state predictor.

\begin{definition}[State Predictor]
    \label{def:ctsp}
    The \emph{state predictor} is $\varepsilon_{\hat{\theta}}(x) \triangleq \dot{x}-\dot{x}_{\hat{\theta}}$ where $\dot{x}_{\hat{\theta}} $ is the instantaneous velocity vector with the current parameter estimate, i.e., $  \dot{x}_{\hat{\theta}} = f(x) - \Delta(x)^\top \hat{\theta} + g(x)u$.
\end{definition}

An important property of the state predictor is that it $\varepsilon_{\hat{\theta}}$ can be written as $\varepsilon_{\hat{\theta}} = \Delta(x)^\top \tilde{\theta}$.
If the state velocities are not directly available, one can use a filtered velocity generated by a first order filter; see \cite{lopez2022adaptive} for details.
The following theorem shows that set invariance is still maintained with the state predictor and direct adaptation law \cref{eq:adapt}.

\begin{theorem}
\label{thm:composite}
Let $\mathcal{C}_{\hat{\theta}}$ be a 0-superlevel set of a continuously differentiable function $h_{\hat{\theta}}:\mathbb{R}^n \times \mathbb{R}^p\rightarrow\mathbb{R}$. If $h_{\hat{\theta}}(x)$  is an uCBF on $\mathcal{C}^r_{\hat{\theta}} \triangleq \{x \in \mathbb{R}^n,  \hat{\theta} \in \Theta : h_{\hat{\theta}}(x) \geq \tfrac{1}{2\gamma} \tilde{\vartheta}^\top \tilde{\vartheta}\} \subseteq \mathcal{C}_{\hat{\theta}}$ then any locally Lipschitz continuous controller satisfying
\begin{equation}
    \underset{u  \in  \mathbb{U}}{\text{sup}}~ \left\{ \nabla_x h_{\hat{\theta}}(x)^\top \left[ f(x) - \Delta(x)^\top \hat{\theta} + g(x) u\right]\right\} \geq - \alpha \left(h_{\hat{\theta}}(x) - \tfrac{1}{2\gamma} \tilde{\vartheta}^\top \tilde{\vartheta}\right) \tag{\ref{eq:rucbfu}}
\end{equation}
renders $\mathcal{C}_{\hat{\theta}}$ safe with the \emph{composite adaptation law}
\begin{subequations}
\label{eq:comp_adapt}
    \begin{align}
        \dot{\hat{\theta}} &= \gamma v(\rho) \Delta(x) \nabla_x h_{\hat{\theta}}(x) - \beta \Delta(x) \varepsilon_{\hat{\theta}} \label{eq:theta_comp} \\
        \dot{\rho} &= - \frac{v(\rho)}{\nabla v(\rho)} \frac{1}{h_{\hat{\theta}}(x)+\eta} \nabla_{\hat{\theta}} h_{\hat{\theta}}(x)^\top  \dot{\hat{\theta}} \label{eq:rho_comp}
    \end{align}
\end{subequations}
where $\gamma$ is an admissible adaptation gain, $v(\rho)$ is scaling function, $\beta \in \mathbb{R}_{>0}$ is the model estimation gain, $\varepsilon_{\hat{\theta}}$ is the state predictor, and $\eta \in \mathbb{R}_{>0}$ is a design parameter.
\end{theorem}
\begin{proof}
Follows similarly to \cref{thm:direct}.
Differentiating $h(t) = v(\rho)(h_{\hat{\theta}}(x) + \eta) - \tfrac{1}{2\gamma}\tilde{\theta}^\top \tilde{\theta}$ and applying \cref{eq:comp_adapt,eq:rucbfu} yields,
\begin{equation*}
    \begin{aligned}
        \dot{h}(t) & \geq - \alpha(h(t) - v(\rho)\eta) + \tfrac{\beta}{\gamma} \tilde{\theta}^\top \Delta(x) \varepsilon_{\hat{\theta}} \\
        & = - \alpha(h(t) - v(\rho)\eta) + \tfrac{\beta}{\gamma} \tilde{\theta}^\top \Delta(x) \Delta(x)^\top \tilde{\theta} \\
        & \geq - \alpha(h(t) - v(\rho)\eta),
    \end{aligned}
\end{equation*}
which is the same inequality obtained in \cref{thm:direct}. Therefore, the controller \cref{eq:rucbfu} and composite adaptation law \cref{eq:comp_adapt} render the set $\mathcal{C}_{\hat{\theta}}$ forward invariant.  
\end{proof}

\subsection{Data-Driven Safety}

One of the key ideas discussed in \cite{lopez2020robust} is the benefit of using a history of data, i.e., life-long model estimation, to reduce the conservatism of adaptive safety controllers.  
Conceptually, the controller in \cref{eq:ucbfu} is trying to render the tightened set  $\mathcal{C}^r_{\hat{\theta}}$  --- not the actual safe set  $\mathcal{C}_{\hat{\theta}}$  --- invariant leading to conservatism. 
Reducing the parameter estimation error bounds  $\tilde{\vartheta}$  via least squares, set membership identification, concurrent learning, Bayesian estimation, etc.~can significantly improve the performance of the closed-loop system since  $\mathcal{C}^r_{\hat{\theta}} \rightarrow \mathcal{C}_{\hat{\theta}}$  as  $\tilde{\vartheta}\rightarrow 0$.
We will show the benefits of life-long model estimation is also applicable to the uCBF adaptive safety framework.
First, we establish a useful lemma.

\begin{lemma}
\label{lemma:tight}
Let $\mathcal{C}^\sigma_\theta$ be a $\sigma$-superlevel set for a continuously differentiable function $h_\theta:\mathbb{R}^n\times \mathbb{R}^p\rightarrow\mathbb{R}$, i.e., $\mathcal{C}^\sigma_\theta \triangleq \{ x \in \mathbb{R}^n,  \theta \in \Theta : h_\theta(x) \geq \sigma\}$ where $\sigma \geq 0$.
If $h_\theta(x)$ is an uCBF on $\mathcal{C}^\sigma_\theta$ then it is also an uCBF on $\mathcal{C}_\theta \supseteq \mathcal{C}^\sigma_\theta$ where $\mathcal{C}_\theta$ is the 0-superlevel set of $h_\theta(x)$.
\end{lemma}
\begin{proof}
If $h_\theta(x)$ is an uCBF on $\mathcal{C}^\sigma_\theta$ then there exists a controller $u$ and extended class $\mathcal{K}_\infty$ function $\alpha(\cdot)$ such that $\dot{h}_\theta(x) \geq -\alpha(h_\theta(x) - \sigma)$. 
Since $\sigma \geq 0$ then $\dot{h}_\theta(x) \geq -\alpha(h_\theta(x))$.
Therefore, $h_\theta(x)$ is also a valid uCBF on $\mathcal{C}_\theta$.
\end{proof}

\begin{theorem}
\label{thm:data}
Let $\mathcal{C}_{\hat{\theta}}$ be a 0-superlevel set of a continuously differentiable function $h_{\hat{\theta}}:\mathbb{R}^n \times \mathbb{R}^p\rightarrow\mathbb{R}$. 
If the model uncertainty $\tilde{\vartheta}(t)$ monotonically decreases via a suitable model estimator and
$h_{\hat{\theta}}(x)$  is an uCBF on $\mathcal{C}^r_{\hat{\theta}} \triangleq \{x \in \mathbb{R}^n,  \hat{\theta} \in \Theta(0) : h_{\hat{\theta}}(x) \geq \tfrac{1}{2\gamma} \tilde{\vartheta}(0)^\top \tilde{\vartheta}(0)\} \subseteq \mathcal{C}_{\hat{\theta}}$ then any locally Lipschitz continuous controller satisfying
\begin{equation}
\label{eq:rucbfut}
    \underset{u  \in  \mathbb{U}}{\mathrm{sup}}~ \left\{ \nabla_x h_{\hat{\theta}}(x)^\top \left[ f(x) - \Delta(x)^\top \hat{\theta} + g(x) u\right]\right\} \geq - \alpha\left(h_{\hat{\theta}}(x) - \tfrac{1}{2\gamma} \tilde{\vartheta}(t)^\top \tilde{\vartheta}(t)\right).
\end{equation}
renders $\mathcal{C}_{\hat{\theta}}$ safe with \cref{eq:adapt} or \cref{eq:comp_adapt} and the suitable model estimator.
\end{theorem}
\begin{proof}
Let $\mathcal{C}^r_{\hat{\theta}}(t) \triangleq \{ x \in \mathbb{R}^n,  \hat{\theta} \in \Theta(t) : h_{\hat{\theta}}(x) \geq \tfrac{1}{2\gamma}\tilde{\vartheta}(t)^\top \tilde{\vartheta}(t)\}$.
If the model uncertainty monotonically decreases then  $\tilde{\vartheta}(t) \leq \tilde{\vartheta}(0)$ so  $\mathcal{C}^r_{\hat{\theta}} = \mathcal{C}^r_{\hat{\theta}}(0) \subseteq  \mathcal{C}^r_{\hat{\theta}}(t)$  for all $t>0$. 
If  $h_{\hat{\theta}}(x)$  is an uCBF on  $\mathcal{C}^r_{\hat{\theta}}$  then it is also an uCBF on  $\mathcal{C}^r_{\hat{\theta}}(t)$  by \cref{lemma:tight}.
Forward invariance of  $\mathcal{C}_{\hat{\theta}}$  then follows the arguments in \cref{thm:direct} or \cref{thm:composite}
\end{proof}

\subsection{Safe Tracking Control}
An unmatched CBF can be immediately combined with an unmatched CLF \cite{lopez2021universal} for safe stabilizing controller  $\kappa: \mathbb{R}^n \times \mathbb{R}^p \times \mathbb{R}^p \rightarrow \mathbb{U}$  that depends on parameter estimates  $\hat{\phi}$  and  $\hat{\theta}$  computed for tracking and safety, respectively. 
The pointwise min-norm controller can be compute by solving the well-known quadratic program
\begin{align*}
    \kappa\left(x,\hat{\phi},\hat{\theta}\right)   =  & ~\underset{u  \in  \mathbb{U}}{\argmin} ~ \frac{1}{2} u^\top u + r \delta^2 \\
    &\mathrm{s.t.} ~ ~\nabla_x V_{\hat{\phi}}(x)^\top \left[ f(x) - \Delta(x)^\top \hat{\phi} + g(x)u \right] \leq - Q_{\hat{\phi}}(x)+ \delta \nonumber  \\
    & \hphantom{s.t.} ~ ~ \nabla_x h_{\hat{\theta}}(x)^\top \left[ f(x) - \Delta(x)^\top \hat{\theta} + g(x) u\right] \geq - \alpha\left(h_{\hat{\theta}}(x) - \tfrac{1}{2\gamma} \tilde{\vartheta}^\top \tilde{\vartheta}\right) \nonumber
\end{align*}
where
\begin{align*}
    \dot{\hat{\phi}} & = - \gamma_c v_c(\varrho)  \Delta(x) \nabla_x V_{\hat{\phi}}(x), \\ 
    \dot{\varrho} & = - \frac{v_c(\varrho)}{\nabla v_c(\varrho)} \frac{1}{V_{\hat{\phi}}(x) + \eta_c} \nabla_{\hat{\phi}} V_{\hat{\phi}}(x)^\top  \dot{\hat{\phi}},
\end{align*}
and 
\begin{align*}
    \dot{\hat{\theta}} &= \gamma_b v_b(\rho) \Delta(x) \nabla_x h_{\hat{\theta}}(x), \\
    \dot{\rho} &= - \frac{v_b(\rho)}{\nabla v_b(\rho)} \frac{1}{h_{\hat{\theta}}(x)+\eta_b} \nabla_{\hat{\theta}} h_{\hat{\theta}}(x)^\top  \dot{\hat{\theta}}.
\end{align*}
Note any of the modifications stated previously can be applied to both the tracking and safety adaptation laws.
One could also formulate a quadratic program for a system that already has a well-designed tracking controller \cite{ames2019control}.

\section{Extension: Adaptive Safety with High Relative Degree Constraints}
The unmatched control barrier function presented in \cref{def:ucbf} can be extended to safety constraints that are more than one derivative away from the input, i.e., those with a relative degree greater than one.
Several works have addressed high relative degree constraints via input-output linearization with known system dynamics \cite{nguyen2016exponential,xiao2021high,tan2021high}.
Recently, \cite{cohen2022high} extended \cite{taylor2020adaptive,lopez2020robust} to high relative degree control barrier functions to uncertain systems, but only for those with uncertainties that satisfy the so-called matching condition.
This section will show that the results in \cref{sec:results} also apply to model-dependent safety constraints with an arbitrary relative degree. 

Now let $\mathcal{C}_\theta$ be a family of 0-superlevel sets parameterized by $\theta$ for a continuously differentiable function $h_\theta:\mathbb{R}^n\times \mathbb{R}^p\rightarrow\mathbb{R}$ which has a well-defined relative degree of $r_b$.
Using the shorthand notation $F_{\theta}(x) \triangleq f_{\theta}(x) - \Delta(x)^\top \theta$, differentiating $h_{\theta}(x)$ until $u$ appears yields ${h}^{(r_b)}_{\theta}(x) = L_{F_\theta}^{r_b} h_{\theta}(x) + L_g L_{F_\theta}^{r_b-1}h_{\theta}(x) u$.
The input-output dynamics were obtained by treating $\theta$ as if it were known in accordance to the certainty equivalence design philosophy. 
Conversely, if $\theta$ were treated as unknown and replaced by $\hat{\theta}$ then higher order derivatives of $\hat{\theta}$ would appear in ${h}^{(r_b)}_{\theta}(x)$ making the design of an adaptive safety controller substantially more difficult.

Existing works \cite{nguyen2016exponential,xiao2021high,tan2021high} use pole placement to stabilize the input-output dynamics yielding forward invariance for known systems.
The work by \cite{cohen2022high} uses the pole placement approach developed in \cite{xiao2021high} for systems with matched uncertainties.
We instead employ a sliding variable $s_\theta(x) \triangleq h^{(r_b-1)}_{\theta}(x) + \phi(h_{\theta}(x),\dots,h^{(r_b-2)}_{\theta}(x))$, which can be viewed as an input into contracting \cite{lohmiller1998contraction} dynamics $\phi : \mathbb{R} \times \dots \times \mathbb{R} \rightarrow \mathbb{R}$.
The sliding variable approach can be viewed as a generalization of the pole placement technique as the contracting dynamics can be designed to have eigenvalues that are a function of $h_{\theta}(x)$, i.e., $\lambda = \lambda(h_\theta(x))$, which can improve time response characteristics.
Examples of suitable sliding variables for $r_b=2$ are $s_{\theta} = \dot{h}_{\theta}(x) + \lambda_1 h_{\theta}(x)$ and $s_{\theta} = \dot{h}_{\theta}(x) + (\lambda_1 + \lambda_2 |h_{\theta}(x)|^q) h_{\theta}(x)$ with $\lambda_i,\,q \in \mathbb{R}_{>0}$, where the latter has a state-dependent eigenvalue $\lambda_1 + \lambda_2 |h_{\theta}(x)|^q$.

\begin{definition}[High-Order Unmatched Control Barrier Function]
\label{def:houcbf}
Let $\mathcal{C}_\theta$ be a family of 0-superlevel sets parameterized by $\theta$ for a continuously differentiable function $h_\theta:\mathbb{R}^n\times \mathbb{R}^p\rightarrow\mathbb{R}$ which has a relative degree of $r_b$.
Additionally, let $s_\theta(x) = h^{(r_b-1)}_{\theta}(x) + \phi(h_{\theta}(x),\dots,h^{(r_b-2)}_{\theta}(x))$ be the input to contracting dynamics given by $\phi : \mathbb{R} \times \dots \times \mathbb{R} \rightarrow \mathbb{R}$.
The function $h_\theta(x)$ is a \emph{high-order unmatched control barrier function} (HOuCBF) if there exists a controller $u \in \mathbb{U}$ and extended class $\mathcal{K}_\infty$ function $\alpha(\cdot)$ such that for every $\theta\in\Theta$

\begin{equation}
\label{eq:houcbfu}
    \underset{u  \in  \mathbb{U}}{\text{sup}}~ \left\{ \nabla_x s_\theta(x)^\top \left[ f(x) - \Delta(x)^\top \theta + g(x) u\right]\right\} \geq - \alpha(s_\theta(x)).
\end{equation}
\end{definition}

Fundamentally, \cref{def:houcbf} states there exists a controller $u$ and extended class $\mathcal{K}_{\infty}$ function $\alpha(\cdot)$ such that $s_{\theta}(x(t)) \geq 0$ uniformly.
The implication $s_{\theta}(x(t)) \geq 0 \implies h_{\theta}(x(t)) \geq 0$ for all $t\geq 0$ can be established by requiring $h^{(i)}_\theta(x_0) \in \mathcal{C}^i_\theta$ for $i=0,\dots,r_b-1$ where $\mathcal{C}^i_\theta$ depends on the choice of contracting dynamics $\phi(\cdot)$.
For example, if $\phi(\cdot)$ has the repeated (constant) eigenvalue $\lambda$ then $\mathcal{C}^i_{\hat{\theta}} = \{ x_0 \in\mathbb{R}^n,\,\theta \in \Theta : \left(\tfrac{d}{dt} + \lambda\right)^i h_{\theta}(x)|_{x_0} \geq 0 \}$ which results in $s_{\theta}(x(t)) \geq 0 \implies h_{\theta}(x(t)) \geq 0$ for all $t\geq 0$.
The following theorem shows a safe set with a high relative degree can be rendered forward invariant using an HOuCBF and direct adaptive control.

\begin{theorem}
\label{thm:hodirect}
Let $\mathcal{C}_{\hat{\theta}}$ be a 0-superlevel set of a continuously differentiable function $h_{\hat{\theta}}:\mathbb{R}^n \times \mathbb{R}^p\rightarrow\mathbb{R}$.
If $h_{\hat{\theta}}(x)$ is a HOuCBF on $\mathcal{C}^r_{\hat{\theta}} \triangleq \{x \in \mathbb{R}^n,  \hat{\theta} \in \Theta : h_{\hat{\theta}}(x) \geq \tfrac{1}{2\gamma} \tilde{\vartheta}^\top \tilde{\vartheta}\} \subseteq \mathcal{C}_{\hat{\theta}}$ with sliding variable $s_{\hat{\theta}}(x)$ such that $s_{\hat{\theta}}(x(t)) \geq 0 \implies h_{\hat{\theta}}(x(t)) \geq 0$ uniformly with $h^{(i)}_{\hat{\theta}}(x_0) \in \mathcal{C}^i_{\hat{\theta}}$ for $i=0,\dots,r_b-1$, then any locally Lipschitz continuous controller satisfying
\begin{equation}
\label{eq:horucbfu}
    \underset{u  \in  \mathbb{U}}{\mathrm{sup}}~ \left\{ \nabla_x s_{\hat{\theta}}(x)^\top \left[ f(x) - \Delta(x)^\top \hat{\theta} + g(x) u\right]\right\} \geq - \alpha\left(s_{\hat{\theta}}(x) - \tfrac{1}{2\gamma} \tilde{\vartheta}^\top \tilde{\vartheta}\right).
\end{equation}
renders $\mathcal{C}_{\hat{\theta}}$ safe with the adaptation law
\begin{subequations}
\label{eq:hoadapt}
    \begin{align}
        \dot{\hat{\theta}} &= \gamma v(\rho) \Delta(x) \nabla_x s_{\hat{\theta}}(x) \label{eq:ho_theta_direct} \\
        \dot{\rho} &= - \frac{v(\rho)}{\nabla v(\rho)} \frac{1}{s_{\hat{\theta}}(x)+\eta} \nabla_{\hat{\theta}} s_{\hat{\theta}}(x)^\top  \dot{\hat{\theta}} \label{eq:ho_rho_direct}
    \end{align}
\end{subequations}
where $\gamma$ is an admissible adaptation gain, $v(\rho)$ is a scaling function, and $\eta \in \mathbb{R}_{>0}$ is a design parameter.
\end{theorem}

\begin{proof}
Follows nearly identical steps to \cref{thm:direct} using the barrier-like function $h(t) = v(\rho) \left(s_{\hat{\theta}}(x) + \eta \right) - \tfrac{1}{2\gamma} \tilde{\theta}^\top \tilde{\theta}$ where $\eta > 0$.
Differentiating $h(t)$ and applying \cref{eq:hoadapt} ultimately yields $\dot{h}(t) \geq - \alpha \left( h(t) - v(\rho)\eta \right)$ so  $h(t) \geq v(\rho) \eta > 0$  for all  $t \geq 0$  if $h(0) \geq v(\rho(0)) \eta$.
Since  $v(\rho) \eta \leq h(t) \leq v(\rho) \left( s_{\hat{\theta}}(x) + \eta\right)$, then  $s_{\hat{\theta}}(x(t)) \geq 0$  for all $t$.
By construction of the sliding variable $s_{\hat{\theta}}(x)$, since $s_{\hat{\theta}}(x(t)) \geq 0$  for all $t$ then $h_{\hat{\theta}}(x(t)) \geq 0$ uniformly.
Therefore, the controller \cref{eq:horucbfu} and direct adaptation law \cref{eq:hoadapt} make the set  $\mathcal{C}_{\hat{\theta}}$  forward invariant. 
\end{proof}

\section{Concluding Remarks}
A new adaptive safety framework was presented that permits the use of the certainty equivalence principle for systematic online selection of a controller that renders a safe set forward invariant despite the presence of unmatched parametric uncertainties. 
The safe combination of policy selection and direct parameter adaptation was achieved by online adjustment of the adaptation gain (inspired by \cite{lopez2021universal}).
The ability to employ the certainty equivalence principle significantly reduces the complexities associated with existing adaptive safety approaches without sacrificing strong theoretical guarantees.
Several modifications that build upon the developed direct adaptation law were shown to also guarantee forward invariance.
Numerous avenues for future work are of interest, many of which center around experimental verification and practical modifications for real-world deployment.
Investigating the viability of data-driven or learning-based techniques to systematically synthesize a family of barrier functions is also of interest.
While the certainty equivalence principle allows for parallelization of this process, there are several fundamental questions on scalability, certifiability, and generalizability when generating barrier (and Lyapunov) functions for uncertain high-dimensional systems.

\section{Appendix}

\begin{definition}[Bregman Divergence]
\label{def:bregman}
Let $\psi(\cdot)$ be a strictly convex, continuously differentiable function on a closed convex set. 
The \emph{Bregman divergence} associated with $\psi(\cdot)$ is given by 
\begin{equation}
\label{eq:bregman}
    \mathrm{d}_\psi ( y ~\|~ x) = \psi(y) - \psi(x) - (y-x)^\top \nabla \psi(x),
\end{equation}
Its time-derivative satisfies $ \ \dot{\mathrm{d}}_\psi (y~\|~x) = (x-y)^\top \nabla^2 \psi(x) \, \dot{x}$ .
\end{definition}

\begin{proposition}
\label{prop:asym_stable}
Let  $h^r_{\hat{\theta}}(x)$  be an RaCBF on the set  $\mathcal{C}^r_{\hat{\theta}} = \{ x \in \mathbb{R}^n,  \hat{\theta}\in\Theta: h^r_{\hat{\theta}}(x) \geq \tfrac{1}{2\gamma} \tilde{\vartheta}^\top \tilde{\vartheta} \}$.
The set $\mathcal{C}^r_{\hat{\theta}}$ is asymptotically with the controller and adaptation law in \cref{thm:racbf}. 
\end{proposition}
\begin{proof}
Consider the Lyapunov-like function
\begin{equation*}
    V_c(t) = \left( \tfrac{1}{2\gamma} \tilde{\vartheta}^\top \tilde{\vartheta} - h^r_{\hat{\theta}}(x) \right) + \tfrac{1}{2\gamma} \tilde{\theta}^\top \tilde{\theta},
\end{equation*}
where $h^r_{\hat{\theta}}(x) \leq \tfrac{1}{2\gamma} \tilde{\vartheta}^\top \tilde{\vartheta}$.
Differentiating and using \cref{eq:racbfu,eq:radapt} yields $\dot{V}_c \leq -\alpha\left(\tfrac{1}{2\gamma} \tilde{\vartheta}^\top \tilde{\vartheta}-h^r_{\hat{\theta}}(x)\right) \leq 0$ so $h^r_{\hat{\theta}}(x)$ and $\tilde{\theta}$ are bounded. 
If $h^r_{\hat{\theta}}(x)$ and $\alpha(\cdot)$ are continuously differentiable functions then $\alpha\left(\tfrac{1}{2\gamma} \tilde{\vartheta}^\top \tilde{\vartheta}-h^r_{\hat{\theta}}(x)\right)$ is uniformly continuous.
Integrating $\dot{V}_c(t)$ yields $\int \limits_0^\infty \alpha\left(\tfrac{1}{2\gamma} \tilde{\vartheta}^\top \tilde{\vartheta}-h^r_{\hat{\theta}}(x(\tau)\right)  d\tau \leq V_c(0) < \infty$ so by Barbalat's lemma, $\alpha(\cdot) \rightarrow 0$. Since $\alpha(r) = 0 \iff r = 0$ then $h^r_{\hat{\theta}}(x(t) \rightarrow \tfrac{1}{2\gamma} \tilde{\vartheta}^\top \tilde{\vartheta}$.
Hence, the set $C^r_{\hat{\theta}}$ is asymptotically stable.
\end{proof}

\bibliographystyle{ieeetr}
\bibliography{ref}

\end{document}

%% file: definitions.tex
\usepackage[pdftex, pdfstartview={FitV}, pdfpagelayout={TwoColumnLeft},bookmarksopen=true,plainpages = false, colorlinks=true, linkcolor=black, citecolor = black, urlcolor = black,filecolor=black , pagebackref=false,hypertexnames=false, plainpages=false, pdfpagelabels ]{hyperref}

\usepackage[authormarkuptext=name,addedmarkup=bf,authormarkupposition=right]{changes}
\definechangesauthor[name={BTL}, color={blue}]{bl}

\usepackage{setspace}
\usepackage{graphicx}
\usepackage{epstopdf}
\usepackage{amssymb,amsmath}

\usepackage{amsthm}
\usepackage{mathtools}
\usepackage{cuted}
\usepackage[font=footnotesize]{subcaption}

\usepackage[font=footnotesize]{caption}
\usepackage{xcolor}
\usepackage{units}
\usepackage{algorithm}
\usepackage[noend]{algpseudocode}
\usepackage{balance}
\usepackage[sort,compress,noadjust]{cite}
\usepackage{tabularx}
\usepackage{nameref}
\usepackage{centernot}

\newtheoremstyle{definition}{}{}{}{}{\bfseries}{.}{.5em}{\thmname{#1}\thmnumber{ #2}\thmnote{ (#3)}}
\theoremstyle{definition}
\newtheorem{definition}{Definition}

\usepackage{enumitem}

\theoremstyle{plain}
\newtheorem{theorem}{Theorem}
\theoremstyle{plain}
\newtheorem{proposition}{Proposition}
\theoremstyle{plain}
\newtheorem{lemma}{Lemma}
\theoremstyle{plain}
\newtheorem{corollary}{Corollary}
\theoremstyle{definition}
\newtheorem{assumption}{Assumption}
\theoremstyle{remark}
\newtheorem{remark}{Remark}

\usepackage[hang,flushmargin]{footmisc}

\usepackage[capitalize]{cleveref}
\crefformat{equation}{(#2#1#3)}
\Crefformat{equation}{Equation~(#2#1#3)}
\Crefname{equation}{Equation}{Eqs.}

\usepackage{accents}

\DeclareMathOperator*{\argmin}{arg\,min}